\documentclass[a4]{article}
\usepackage{a4wide}
\usepackage[applemac]{inputenc}

\usepackage[english]{babel}
\usepackage{latexsym}
\usepackage{amssymb}
\usepackage{amsmath}
\usepackage{amsthm}
\usepackage{ifthen}
\usepackage{xspace}
\usepackage{cite}      
\usepackage{enumerate}
\usepackage{subfigure}
\usepackage{tikz}
\usepackage[colorlinks,final,linkcolor=blue,citecolor=blue,urlcolor=blue]{hyperref}

\renewcommand{\phi}{\varphi}
\newcommand{\eps}{\varepsilon}

\newcommand{\set}[2]{\left\{#1\mathrel{\left|\vphantom{#1}\vphantom{#2}\right.}#2\right\}}
\newcommand{\oneset}[1]{\left\{\mathinner{#1}\right\}}

\newcommand{\abs}[1]{\left|\mathinner{#1}\right|}

\newcommand{\ubr}[2]{\underbrace{#1\vphantom{g_g}}_{#2}}

\newcommand{\N}{\mathbb{N}}

\newcommand{\Oh}{\mathcal{O}}

\newcommand{\cC}{\mathcal{C}}

\newcommand{\cH}{\mathcal{H}}
\newcommand{\cL}{\mathcal{L}}
\newcommand{\cP}{\mathcal{P}}
\newcommand{\cR}{\mathcal{R}}

\newcommand{\e}{\eps}

\newcommand{\NP}{\ensuremath{\mathrm{NP}}\xspace}
\newcommand{\NL}{\ensuremath{\mathrm{NL}}\xspace}

\newcommand{\Hk}{\cH_k}
\newcommand{\Hks}{\cH_k^*}
\newcommand{\Ha}{\cH_\alp}
\newcommand{\Has}{\cH_\alp^*}

\newcommand{\aL}{\cL_\alp}
\newcommand{\aR}{\cR_\alp}

\newcommand{\Pa}{\cP_\alp}

\newcommand{\oneH}{\cR\!\Hk}
\newcommand{\twoH}{\Hk}

\newcommand{\tHw}{\ensuremath{\twoH^*(\oneset{w})}\xspace}
\newcommand{\oHw}{\ensuremath{\oneH^*(\oneset{w})}\xspace}

\newcommand{\BAR}{\overline{\phantom{ii}}}
\newcommand{\smalloverline}[1]
	{{\mspace{1mu}\overline{\mspace{-1mu}#1\mspace{-1mu}}\mspace{1mu}}}
\newcommand{\ov}[1]{\smalloverline{#1}}

\newcommand{\HCS}{\ensuremath{\mathrm{HCS}}\xspace}
\newcommand{\HCSk}{\ensuremath{\mathrm{HCS}_k}\xspace}

\newcommand{\hpc}{hairpin completion\xspace}
\newcommand{\hpcs}{hairpin completions\xspace}
\newcommand{\phpc}{parameterized \hpc}

\newcommand{\iphpc}{iterated \phpc}
\newcommand{\ihpc}{iterated \hpc}
\newcommand{\ihpcs}{iterated \hpcs}
\newcommand{\bhpc}{bounded \hpc}
\newcommand{\ibhpc}{iterated \bhpc}

\newcommand{\alp}{\alpha}
\newcommand{\gam}{\gamma}
\newcommand{\Sig}{\Sigma}
\newcommand{\sse}{\subseteq}

\newcommand{\gabag}{\gamma\alpha\beta\ov\alpha\ov\gamma}
\newcommand{\gaba}{\gamma\alpha\beta\ov\alpha}
\newcommand{\abag}{\alpha\beta\ov\alpha\ov\gamma}

\newtheorem{theorem}{Theorem}[section]
\newtheorem{lemma}[theorem]{Lemma}
\newtheorem{corollary}[theorem]{Corollary}

\newenvironment{example}
	{\begin{trivlist}\item[\hskip\labelsep {\itshape Example.\,}]}
	{\end{trivlist}}

\begin{document}

\title{On the Iterated Hairpin Completion}
\author 
	{Steffen Kopecki \\
	{\tt kopecki@fmi.uni-stuttgart.de} \\
	\small
	University of Stuttgart,
	Institute for Formal Methods in Computer Science (FMI), \\
	\small
	Universit\"atsstra\ss e 38,
	D-70569 Stuttgart}


\maketitle

\begin{abstract}
	The \emph{(bounded) \hpc }and its iterated versions
	are operations on formal languages
	which have been inspired by the hairpin formation in DNA-biochemistry. 
	The paper answers two questions asked in the literature about the iterated \hpc.
		
	The first question is whether the class of regular languages is closed under
	iterated bounded hairpin completion.
	Here we show that this is true by providing a more general result which applies to all
	the classes of languages 
	which are closed under finite union, intersection with regular sets,
	and concatenation with regular sets.
	In particular, all Chomsky classes and all 
	standard complexity
	classes are closed under iterated bounded hairpin completion.
	
	In the second part of the paper we address the question whether
	the iterated \hpc of a singleton is always regular. 
	In contrast to the first question, this one has a negative answer.
	We exhibit an example of a singleton language whose
	iterated \hpc is not regular, actually it is not context-free, 
	but context-sensitive.
\end{abstract}

\begin{center}
\end{center}

\section{Introduction}

The hairpin completion is an operation on formal languages which is inspired by
DNA-comput\-ing and biochemistry
where it appears naturally in chemical reactions. It turned out that the corresponding operation on 
formal languages gives rise to very interesting and quite subtle 
decidability and computational problems.
The focus of this paper is on these formal language theoretical results.
However, let us sketch the biochemical origin of this operation first.

A {\em DNA strand} is a polymer composed of nucleotides
which differ from each other by their bases $A$ (adenine), $C$ (cytosine), $G$ (guanine), and $T$ (thymine).
For our purposes a strand can be seen as a finite sequence of bases.
By {\em Watson-Crick base pairing} two base sequences can bind to each other if they are pairwise complementary,
where $A$ is complementary to $T$ and $C$ to $G$.
The \hpc is best explained by Figure~\ref{fig:hpc}.
By a sequence $\ov w$ we always mean to read $w$ from right to left and to
complement base by base, i.e., $\ov{a_1 \cdots a_n} =  \ov{a_n} \cdots \ov {a_1}$.
During a chemical process, called annealing,
a strand which contains a sequence $\alp$ and ends on the complementary sequence
$\ov\alp$, Fig.~\ref{fig:strand},
can form an intramolecular base-pairing which is known as {\em hairpin}
(in case $\alp$ is not too short, say $\abs \alp\geq 10$), see Fig.~\ref{fig:hairpin}.
By complementing the unbound sequence $\gamma$, the {\em \hpc} arises, Fig.~\ref{fig:completion}.

\begin{figure}[h]
	\centering
	\subfigure[{strand}]
	{\label{fig:strand}
		\begin{tikzpicture}
			\draw [|-|] (0,0) .. controls +(30:.5) and +(210:.5) ..
				node [above,sloped] {$\gamma$} (1,0);
			\draw [-] (1,0) .. controls +(30:.25) and +(150:.25) .. 
				node [above,sloped] {$\alpha$} (1.5,0);
			\draw [|-|] (1.5,0) .. controls +(330:.5) and +(150:.5) ..
				node [above,sloped] {$\beta$} (2.5,0);
			\draw [-|] (2.5,0) .. controls +(330:.25) and +(210:.25) .. 
				node [above,sloped] {$\vphantom\gamma\ov\alpha$} (3,0);
			\node at (-.5,-.7) {};
			\node at (3.5,0) {};
		\end{tikzpicture}
	}
	\subfigure[hairpin]
	{\label{fig:hairpin}
		\begin{tikzpicture}
			\draw [|-] (0,0) -- node [above] {$\gamma$} (1,0);
			\draw [|-|] (1,0) -- node [above] {$\vphantom\gamma\alpha$} (1.5,0);
			\draw [-] (1.5,0) .. controls +(right:.25) and +(left:.25) .. (2,.25)
				.. controls +(right:.5) and +(right:.5) .. node [right] {$\beta$} (2,-.5)
				.. controls +(left:.25) and +(right:.25) .. (1.5,-.25);
			\draw [|-|] (1.5,-.25) -- node [below] {$\ov\alpha$} (1,-.25);
			\node at (0,-.7) {};
		\end{tikzpicture}	
	}
	\subfigure[\hpc]
	{\label{fig:completion}
		\begin{tikzpicture}
			\draw [|-] (0,0) -- node [above] {$\gamma$} (1,0);
			\draw [|-|] (1,0) -- node [above] {$\vphantom\gamma\alpha$} (1.5,0);
			\draw [-] (1.5,0) .. controls +(right:.25) and +(left:.25) .. (2,.25)
				.. controls +(right:.5) and +(right:.5) .. node [right] {$\beta$} (2,-.5)
				.. controls +(left:.25) and +(right:.25) .. (1.5,-.25);
			\draw [|-|] (1.5,-.25) -- node [below] {$\ov\alpha$} (1,-.25);
			\draw [-|] (1,-.25) -- node [below] {$\ov\gamma$} (0,-.25);
			\node at (-.5,-.7) {};
			\node at (2.9,0) {};
		\end{tikzpicture}	
	}
	\caption{Hairpin completion of a DNA-strand.}
	\label{fig:hpc}
\end{figure}
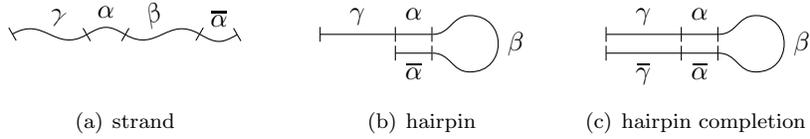

Hairpin completions of strands develop naturally during a
technique called {\em Polymerase Chain Reaction} (PCR).
The PCR is often used in DNA algorithms to amplify DNA strands with certain properties.
In many algorithms which use PCR the hairpin completions are by-products
which cannot be used for the subsequent computation.
Therefore, sets of strands which are unlikely to build hairpins (or lead to other {\em bad hybridizations})
have been examined in many papers, see e.g., \cite{garzon1,garzon2,garzon3,KariKLST05,KariMT07}.

On the other hand, some DNA-based computations rely on the fact that
DNA strands can form hairpins.
An example are algorithms using the {\em Whiplash PCR}
in which strands are designed to build hairpins.
This technique can be used to solve combinatorial problems,
including $\NP$-complete ones like \textsc{Satisfiability} and \textsc{Hamiltonian Path},
see \cite{Hagiya97,Sakamoto98,Winfree98}.

\bigskip

On an abstract level a strand can be seen as a word and a
(possibly infinite) set of strands is a language.
The \hpc of formal languages has been introduced in \cite{ChepteaMM06}
by Cheptea, Mart\'{\i}n-Vide, and Mitrana.
In several papers the \hpc and some familiar operations have been studied, see e.g.,
\cite{ChepteaMM06,ManeaMY10,ManeaMY09tcs,DBLP:conf/cie/ManeaM07,ManeaMM09,DiekertKM09}.
The focus of this paper is on closure properties of language classes
concerning the iterated versions of the \hpc and the \bhpc.
For the latter operation we assume the length of the $\gamma$-part to be bounded.
This variant of the hairpin completion 
was introduced and analyzed in \cite{ItoLM09,Ito2010}
by Ito, Leupold, Manea, and Mitrana.
A formal definition of both operations is given in Section~\ref{sec:def:hairpin}.

In \cite{ChepteaMM06} the closure properties of different
language classes under the non-iterated and iterated hairpin completion
have been analyzed. 
It follows that neither regular nor
context-free languages are closed under hairpin completion
whereas the family of context-sensitive languages
is closed under this function.
Actually, from \cite{ChepteaMM06} we can derive that 
the class $\mathrm{DSPACE}(f)$ (resp.\ the class $\mathrm{NSPACE}(f)$) is closed under  \hpc
(resp.\ closed under iterated \hpc) 
for every function $f\in \Omega(\log)$.
(By the class $\mathrm{DSPACE}(f)$ (resp.\ $\mathrm{NSPACE}(f)$)
we mean, as usual, the class of languages that can be accepted by
a deterministic (resp.\ non-deterministic) {T}uring  
machine which uses $f(n)$ work space on input length $n$.)
In particular, the class of context-sensitive languages is closed under iterated hairpin completion, too.
Furthermore, if we apply the iterated hairpin completion to a regular (resp. context-free) language
we stay inside \NL{}($=\mathrm{NSPACE}(\log)$)
(resp.\ $\mathrm{NSPACE}(\log^2)$, by Lewis, Stearns, and Hartmanis \cite{LewisSH65})
which is in terms of space complexity far below the class of deterministic context-sensitive languages. 

The situation changes if we consider the bounded hairpin completion,
which can be seen as a weaker variant of the hairpin completion.
All classes in the Chomsky Hierarchy are closed under bounded hairpin completion
and the classes of context-free, context-sensitive,
and recursively enumerable languages are closed under
the iterated operation, see \cite{ItoLM09,Ito2010}.
But the status for regular languages remained unknown and was stated as an
open problem in \cite{ItoLM09}.
In Section~\ref{sec:ibhpc} we solve this problem.
We state a general representation for the
\ibhpc of a formal language using the operations
union, intersection with regular sets, and concatenation with regular sets
(Theorem~\ref{main}).
As a consequence all language classes which are closed under these basic
operations are also closed under \ibhpc.

Furthermore, for a given non-deterministic finite automaton (NFA) accepting a language $L$,
we give exponential lower and upper bounds for  the size of an NFA accepting
the iterated bounded hairpin completion of $L$ in Theorem~\ref{thm:nfa}. 
Thus, if we ignore constants, the NFA leads us to a linear time  membership test
for the iterated bounded  hairpin completion of a fixed regular language.
This improves a quadratic bound which was known before. Indeed, 
the best known time complexity of
the membership problem for the iterated (unbounded) hairpin completion of a
regular language $L$ is still quadric time
by an algorithm from \cite{ManeaMM09}. See Section~\ref{sec:size} for a more detailed discussion.

The class of \ihpcs of singletons (\HCS) has been investigated in \cite{ManeaMY10}
by Manea, Mitrana, and Yokomori
(which is the journal version of a paper that appeared at AFL~2008).
Obviously, \HCS is included in the class of context-sensitive languages.
However, the questions if \HCS contains non-regular or non-context-free languages
has been unsolved.
In Section~\ref{sec:ihpc} we answer this question  by stating
a singleton whose iterated hairpin completion is not context-free.

This paper is the journal version of results which appeared as a poster
at DLT~2010, \cite{Kopecki10}.

\section{Definitions and Notation}

We assume the reader to be familiar with the fundamental concepts of formal
language and automata theory, see \cite{HU}.

An \emph{alphabet} is a finite set of \emph{letters}.
In this paper the alphabet is always $\Sigma$. 
The set of words over $\Sigma$ is denoted by $\Sigma^*$, as usual, and the
\emph{empty word} is denoted by $\e$.
We consider $\Sigma$ with an \emph{involution};
this is a bijection $\BAR: \Sigma\to \Sigma$ such that
$\overline{\ov{a}} = a$ for all letters $a \in \Sigma$
(in DNA-biochemistry: $\Sigma = \{A,C,G,T\}$ with $\ov A= T$ and $\ov C= G$).
We extend the involution to words $w = a_1 \cdots a_n $ 
by  $\ov w =  \ov{a_n} \cdots \ov{a_1}$. (Just like taking inverses in groups.)
For a formal language $L$ by $\ov L$ we denote the language $\set{\ov w}{w\in L}$.

Given a word $w$, we denote by $\abs w$ its length.
For a length bound $\ell\geq 0$ the set $\Sig^{\leq\ell}$ contains all words of length at most $\ell$.
If $w=xyz$ for some $x,y,z\in \Sigma^*$, then $x$, $y$, and $z$ are called 
{\em prefix}, {\em factor}, and {\em suffix} of the word $w$, respectively.
For the prefix relation we also use the notation $x \leq w$.
Note that if $z$ is a suffix of $w$, then $\ov z$ is a prefix of $\ov w$ (or $\ov z\leq \ov w$).

A common way to describe regular languages are {\em non-deterministic finite automata} (NFAs).
An NFA $A$ is a tuple $(Q,\Sig,E,I,F)$ where $Q$ is the {\em finite set of states},
$I\sse Q$ is the set of {\em initial states}, $F\sse Q$ is the set of {\em final states}, and
$E \sse Q\times \Sig \times Q$ is the set of {\em labelled edges} or {\em transitions}.
The language accepted by the automaton, denoted by $L(A)$, contains all words $w$
such that there is a path labelled by $w$ which leads from an initial state to a final state.
By the size of an NFA we mean the number of states $\abs Q$.

\subsection{The Hairpin Completion}\label{sec:def:hairpin}

Let $w\in\Sig^*$ be a word. If $w$ has a factorization $w=\gaba$,
it can form a hairpin and $\gabag$ is a {\em right hairpin completion} of $w$
(again, see Figure~\ref{fig:hpc}).
Since a hairpin in biochemistry is stable only if $\alp$ is long enough, we
fix a constant $k\geq 1$ and ask $\abs\alp = k$.
(Note that the definition does not change if we ask $\abs\alp \geq k$.)

Symmetrically, if $w$ has a factorization $\abag$ with $\abs\alp = k$,
then $\gabag$ is a {\em left hairpin completion} of $w$.
For the bounded hairpin completion we assume that the length of the factor $\gam$ is
bounded by some constant.

The hairpin completion of a formal language $L$ is the union of all hairpin completions
of all words in $L$.
Before we state the formal definition of the unbounded and bounded hairpin completion of a language,
we introduce a more general variant of the hairpin completion, namely the
{\em parameterized hairpin completion}.
The parameterized hairpin completion covers the other operations as special cases.

Let $\ell, r\in \N\cup \oneset\infty$ be two length bounds
and let $L$ be a formal language.
Considering a left hairpin completion with the factorization $\gabag$ as above,
then the bound $\ell$ limits the length of $\gamma$;
respectively, the bound $r$ limits the length of $\gamma$ in a right hairpin completion.
For a word $\alp\in \Sig^k$ the parameterized hairpin completion is defined as
\begin{gather*}
	\Ha(L,\ell,0) = \bigcup_{\gamma\in\Sig^{\leq\ell}} \gam \left( \alp\Sig^*\ov\alp \ov\gam \cap L\right) \\
	\Ha(L,0,r) = \bigcup_{\gamma\in\Sig^{\leq r}} \left( \gam\alp\Sig^*\ov\alp \cap L \right) \ov\gam \\
	\Ha(L,\ell,r) = \Ha(L,\ell,0) \cup \Ha(L,0,r).
\end{gather*}
For the constant $k$ we define
\begin{equation*}
	\Hk(L,\ell,r) = \bigcup_{\alp\in\Sig^k} \Ha(L,\ell,r).
\end{equation*}

In the unbounded case we distinguish two operations:
The {\em (two-sided) hairpin completion} is defined as $\twoH(L) = \Hk(L,\infty,\infty)$
and the {\em right-sided hairpin completion} is defined as $\oneH(L) = \Hk(L,0,\infty)$.
For the latter case we allow right hairpin completions, only.
In the same way we might define the left-sided hairpin completion of a 
language,
but for convenience we will treat the right-sided operation, only, and
also refer to it as the {\em one-sided \hpc}.
It is plain, that our results also hold for the left-sided case.

The {\em bounded hairpin completion} $\cH(L,m,m)$ arises if we choose
the same finite bound $m\in\N$ for left and right hairpin completions.

Note that if both bounds $\ell,r$ are finite and $L$ is regular, then
the parameterized hairpin completion $\Hk(L,\ell,r)$ is regular as well.
This does not hold if $\ell=\infty$ or $r=\infty$ as one of the unions becomes infinite.
It is known that the unbounded hairpin completion of a regular language is not necessarily
regular but always linear context-free, see e.g., \cite{ChepteaMM06}.

In this paper we examine the iterated versions of the operations we defined so far.
The iterated hairpin completion of a language $L$ contains all words which belong to a sequence
$w_0,\ldots, w_n$ where $w_0\in L$ and where
$w_i$ is a right or left hairpin completion of $w_{i-1}$ and the bound $r$ (resp.\ $\ell$) applies
for all $i$ such that $1\leq i\leq n$.
More formal, let $\ell,r\in \N\cup\oneset\infty$ and
\begin{align*}
	&  \Ha^0(L,\ell,r) = L,
	&& \Ha^i(L,\ell,r) = \Ha(\Ha^{i-1}(L,\ell,r),\ell,r), \\
	&  \Hk^0(L,\ell,r) = L,
	&& \Hk^i(L,\ell,r) = \Hk(\Hk^{i-1}(L,\ell,r),\ell,r)
\end{align*}
for $i\geq 1$.
The iterated parameterized hairpin completion of $L$ is the union
\begin{align*}
	&   \Has(L,\ell,r) = \bigcup_{i\geq 0} \Ha^i(L,\ell,r)
	&& \text{resp.}
	&& \Hks(L,\ell,r) = \bigcup_{i\geq 0} \Hk^i(L,\ell,r).
\end{align*}

If a word $z$ is included in $\Hk^i(\oneset{w},\ell,r)$,
we say $z$ is an $i$-iterated hairpin completion of $w$,
and if $z\in \Hks(\oneset{w},\ell,r)$, we say $z$ is an iterated hairpin completion of $w$.
(It will be clear from the context which length bounds apply.)

The iterated unbounded hairpin completions are denoted by
$\twoH^*(L) = \Hks(L,\infty,\infty)$ and
$\oneH^*(L) = \Hks(L,0,\infty)$.

\begin{example}
Figure~\ref{fig:ihpc} shows a $3$-iterated hairpin completion of
$\alp u \ov \alp v\alp$ where $\abs\alp = k$. 
In each step the dotted part is the newly created prefix or suffix.

\begin{figure}[h]
	\centering
	\begin{tikzpicture}
		\newcommand{\cut}{node [sloped] {{\tiny$\mid$}}}
		\newcommand{\shortcut}{node [sloped] {{\tiny$\mid$}}}
		\begin{scope}
			\draw												  (0,.1)			\cut
				--	node [above]		{$\alp$}						++(.3,0)			\shortcut
				--	node [above]		{$u$}						++(.5,0)			\cut
				--	node [above]		{$\ov\alp$}					++(.3,0)			\cut
				..	controls +(right:.2) and +(left:.2) ..		++(.3,.15)
					node [above]		{$v$}
				..	controls +(right:.3) and +(right:.3)	 ..		++(0,-.5)
				..	controls +(left:.2) and +(right:.2)..		++(-.3,.15)		\cut
				--	node [below]		{$\alp\vphantom{\ov a}$}		++(-.3,0)		\cut;
			\draw [dotted]										  (0,-.1)		\cut
				--	node [below]		{$\ov\alp$}					++(.3,0)			\shortcut
				--	node [below]		{$\ov u$}					++(.5,0);
		\end{scope}

		\begin{scope}[xshift=3.1cm]
			\draw												  (0,.1)			\cut
				--	node [above]		{$\alp$}						++(.3,0)			\shortcut
				--	node [above]		{$u$}						++(.5,0)			\shortcut
				--	node [above]		{$\ov\alp$}					++(.3,0)			\shortcut
				--	node [above]		{$v$}						++(.5,0)			\cut
				--	node [above]		{$\alp$}						++(.3,0)			\cut
				..	controls +(right:.2) and +(left:.2) ..		++(.3,.15)
					node [above]		{$\ov u$}
				..	controls +(right:.3) and +(right:.3)	 ..		++(0,-.5)
				..	controls +(left:.2) and +(right:.2)..		++(-.3,.15)		\cut
				--	node [below]		{$\ov\alp$}					++(-.3,0)		\cut;
			\draw [dotted]										  (0,-.1)		\cut
				--	node [below]		{$\ov\alp$}					++(.3,0)			\shortcut
				--	node [below]		{$\ov u$}					++(.5,0)			\shortcut
				--	node [below]		{$\alp\vphantom{\ov a}$}		++(.3,0)			\shortcut
				--	node [below]		{$\ov v$}					++(.5,0)	;
		\end{scope}

		\begin{scope}[xshift=7cm]
			\draw												  (0,-.1)		\cut
				--	node [below]		{$\ov\alp$}					++(.3,0)			\shortcut
				--	node [below]		{$\ov u$}					++(.5,0)			\shortcut
				--	node [below]		{$\alp\vphantom{\ov a}$}		++(.3,0)			\shortcut
				--	node [below]		{$\ov v$}					++(.5,0)			\cut
				--	node [below]		{$\ov \alp$}					++(.3,0)			\cut
				..	controls +(right:.2) and +(left:.2) ..		++(.3,-.15)
				--												++(.5,0)
				..	controls +(right:.3) and +(right:.3)	 ..		++(0,.5)
				--	node [above]		{$u\ov\alp v\alp\ov u$}		++(-.5,0)
				..	controls +(left:.2) and +(right:.2)..		++(-.3,-.15)		\cut
				--	node [above]		{$\alp$}						++(-.3,0)		\cut;
			\draw [dotted]										  (0,.1)			\cut
				--	node [above]		{$\alp$}						++(.3,0)			\shortcut
				--	node [above]		{$u$}						++(.5,0)			\shortcut
				--	node [above]		{$\ov \alp$}					++(.3,0)			\shortcut
				--	node [above]		{$v$}						++(.5,0)	;
		\end{scope}
	\end{tikzpicture}
	\caption{Example for the \ihpc.}
	\label{fig:ihpc}
\end{figure}
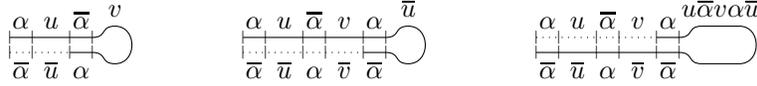
\end{example}

\section{The Iterated Bounded Hairpin Completion}
\label{sec:ibhpc}

In this section we will give a general representation for the \iphpc
with finite bounds.
Our main result is the proof of the following theorem which can be found in Section~\ref{sec:mainproof}.

\begin{theorem}\label{main}
	Let $L$ be a formal language and $\ell,r\in \N$.
	The \iphpc $\Hks(L,\ell,r)$ can be effectively represented by an expression using
	$L$ and the operations union, intersection with regular sets, and concatenation with regular sets.
\end{theorem}

Consequentially, all language classes which are closed under these operations
are also closed under \iphpc with finite bounds,
and if the closure under all three operations is effective, then the closure
under \iphpc with finite bounds is effective, too;
this applies to all four Chomsky classes.
From \cite{ItoLM09} it is known that the classes of context-free, context-sensitive,
and recursively enumerable languages are closed under \ibhpc,
but the status for regular languages was unknown.
Since the \ibhpc is a special case of the \iphpc with finite bounds
we can answer this question now.

\begin{corollary}
	Let $\cC$ be a class of languages.
	If $\cC$ is closed under union, intersection with regular sets,
	and concatenation with regular sets, then $\cC$ is also closed
	under iterated bounded hairpin completion.
	Moreover, if $\cC$ is effectively closed under union, intersection with regular sets,
	and concatenation with regular sets, then
	the closure under iterated bounded hairpin completion is effective.
	
	In particular, the class of regular languages is effectively closed under \ibhpc.
\end{corollary}

The next two sections are devoted to the proof of Theorem~\ref{main}.
First we introduce the important concept of {\em $\alp$-prefixes}.

\subsection{$\alp$-Prefixes}

Let $\alp$ be a word of length $k$. For $v,w\in\Sigma^*$ we say
$v$ is an $\alp$-prefix of $w$ if $v\alp \leq w$.
We denote the set of all $\alp$-prefixes of length at most $\ell$ by
\begin{equation*}
	\Pa(w,\ell) = \set{v}{v\alp \leq w\wedge \abs v\leq \ell}.
\end{equation*}

The idea behind this notation is: For a word $w\in\alp\Sig^*\ov\alp$ with $\abs w-k \geq \ell,r$,
the set of (non-iterated) parameterized hairpin completions of $w$ is given by
\begin{align*}
	&  \Ha(\oneset{w},\ell,0) = \Pa(\ov w,\ell)w
	&& \text{and}
	&& \Ha(\oneset{w},0,r) = w\ov{\Pa(w,r)}.
\end{align*}

In the following proof we are interested in $\alp$-prefixes of words which have $\alp$ as a prefix.
This leads to some useful properties.

\begin{lemma}\label{lem:prefixes}
	Let $\alp\in\Sig^k$, $\ell\in\N$, and $w\in\alp\Sig^*$.
	\begin{enumerate}
		\item For all $v\in\Pa(w,\ell)$ we have $\alp\leq v\alp$. 
		\item For all $u,v\in\Pa(w,\ell)$ we have
			\begin{equation*}
				\abs{u} \leq \abs{v}	\ \Leftrightarrow\  u\alp \leq v\alp
				\ \Leftrightarrow\  u\in\Pa(v\alp,\ell).
			\end{equation*}
		\item If $v\alp$ is a prefix of some word in $\Pa(w,\ell)^*\alp$, then
			$v\in\Pa(w,\ell)^*$. 
	\end{enumerate}
\end{lemma}

\begin{proof}
	If two words $x$, $y$ are prefixes of $w$ and $\abs x \leq \abs y$, then $x\leq y$.
	This yields properties~1 and 2.
	
	For property~3 let $v\alp \leq x_1\cdots x_m\alp$ where $x_1,\ldots,x_m\in\Pa(w,\ell)$.
	We can factorize $v = x_1\cdots x_{i-1} y$ such that $y\leq x_i$
	for some $i$ with $1\leq i \leq m$. By property~1 and induction, we see that
	$\alp$ is a prefix of $x_{i+1}\cdots x_m\alp$ and hence $y\alp\leq x_i\alp \leq w$
	which implies $y\in\Pa(w,\ell)$ and, moreover, $v\in\Pa(w,\ell)^*$.
\end{proof}

\subsection{Proof of Theorem~\ref{main}}\label{sec:mainproof}

Let $L$ be a formal language and $\ell, r\in \N$. 
We will state a representation for $\Hks(L,\ell,r)$ using $L$ and the operations union,
intersection with regular sets, and concatenation with regular sets.

Let us begin with a basic observation.
Every word $w$ which is a hairpin completion of some other word has
a factorization $w=\delta\beta\ov\delta$ with $\abs\delta \geq k$, therefore,
the prefix of $w$ of length $k$ and the suffix of $w$ of length $k$ are complementary.
Let us call this prefix $\alp$, hence, we have $w\in\alp\Sig^*\ov\alp$.
Every word which is a right hairpin completion of $w$ has still the prefix $\alp$
and since the suffix of length $k$ is complementary, it has the suffix $\ov\alp$ as well.
For left hairpin completions we have a symmetric argument and, by induction,
every word which is an iterated hairpin completion of $w$ has prefix $\alp$ and suffix $\ov\alp$.

Thus, we can split up the (non-iterated) parameterized hairpin completion $\Hk(L,\ell,r)$
into finitely many languages
$L_\alp = \Hk(L,\ell,r) \cap \alp\Sig^*\ov\alp$ where $\alp\in\Sig^k$, and each of them
has a effective representation using $L$ and the operations union, intersection with regular sets, and
concatenation with regular sets.
Moreover,
\begin{equation*}
	\Hks(L_\alp,\ell,r) = \Has(L_\alp,\ell,r) \sse \alp\Sig^*\ov\alp
\end{equation*}
and the iterated parameterized hairpin completion equals
\begin{align*}
	\Hks(L,\ell,r) &= L \cup \Hks(\Hk(L,\ell,r),\ell,r) \\
	&= L \cup \Hks\Bigl(\bigcup_{\alp\in\Sig^k}L_\alp,\ell,r\Bigr) \\
	&= L \cup \bigcup_{\alp\in\Sig^k} \Has(L_\alp,\ell,r).
\end{align*}

Henceforth, let $\alp\in\Sig^k$ be fixed.
In order to prove Theorem~\ref{main} we will state a suitable representation for
$\Has(L_\alp,\ell,r)$.
For the rest of the proof we will heavily rely on the fact that every word in
$\Has(L_\alp,\ell,r)$ has the prefix $\alp$ and the suffix
$\ov\alp$.
The representation is defined recursively.
We have
\begin{equation*}
	\Has(L_\alp,0,0) = L_\alp.
\end{equation*}

By symmetry, we may assume that $\ell\geq r$ and $\ell \geq 1$.
We will state a representation for $\Has(L_\alp,\ell,r)$
using $\Has(L_\alp,\ell-1,r)$ and the operations union, intersection with regular sets, and
concatenation with regular sets.
Therefore, consider a word 
\begin{equation*}
	z \in \Has(L_\alp,\ell,r)\setminus \Has(L_\alp,\ell-1,r).
\end{equation*}
For some $n\geq 1$ there is a sequence $w_0,\ldots,w_n=z$ where $w_0\in L_\alp$
and for all $i$ such that $1\leq i\leq n$
either $w_i$ is a left hairpin completion of $w_{i-1}$ and $\abs{w_i}\leq \abs{w_{i-1}}+\ell$
or $w_i$ is a right hairpin completion of $w_{i-1}$ and $\abs{w_i}\leq \abs{w_{i-1}}+r$.
Furthermore, there is an index $j \geq 1$ such that $w_{j-1} = w \in \Has(L_\alp,\ell-1,r)$
and $w_{j} = vw \notin \Has(L_\alp,\ell-1,r)$.
Note that this implies $\abs v = \ell$ and $w\in\alp\Sig^*\ov\alp\ov v$.
Let $s = n-j$ and consider the factorization
\begin{equation*}
	z = x_s\cdots x_1 v w \ov{y_1} \cdots \ov{y_s}
\end{equation*}
where $x_i\cdots x_1 v w \ov{y_1} \cdots \ov{y_i} = w_{j+i}$ and either
\begin{enumerate}
	\item $y_i = \e$, $\abs{x_i}\leq \ell$, and
		$x_i\alp \leq y_{i-1}\cdots y_1 v\alp$ or 
	\item $x_i = \e$, $\abs{y_i}\leq r$, and
		$y_i\alp \leq x_{i-1}\cdots x_1 v\alp$. 
\end{enumerate}
for all $i$ such that $0\leq i\leq s$.

The crucial point is that $v w$ has the prefix $v\alp$, the suffix $\ov\alp\ov v $,
and $\abs{v} = \ell \geq r$. Therefore, the factors $x_1,\ldots,x_s$ and $y_1,\ldots,y_s$ are
controlled by the triple $(v,\ell,r)$ in the following way.

\begin{lemma}\label{lemon}
	$x_i \in \Pa(v\alp,\ell)^*$ and $y_i \in \Pa(v\alp,r)^*$ for all $i$ such that $1\leq i\leq s$.
\end{lemma}

\begin{proof}
	We prove the claim by induction on $i$.
	Let $i$ such that $1\leq i \leq s$.
	Our induction hypothesis is $x_j \in \Pa(v\alp,\ell)^*$
	and $y_j \in \Pa(v\alp,r)^*$ for all $j$ such that $1\leq j < i$.
	We distinguish between the two cases above:
	
	\begin{enumerate}
		\item We have $y_i = \e \in \Pa(v\alp,r)^*$ and, by induction hypothesis,
			\begin{equation*}
				x_i\alp \leq y_{i-1}\cdots y_{1}v\alp \in \Pa(v\alp,r)^*v\alp \sse \Pa(v\alp,\ell)^*\alp.
			\end{equation*}
			In combination with Lemma~\ref{lem:prefixes} this yields $x_i\in\Pa(v\alp,\ell)^*$.
		\item We have $x_i = \e \in \Pa(v\alp,\ell)^*$ and
			\begin{equation*}
				y_i\alp \leq x_{i-1}\cdots x_{1}v\alp \in \Pa(v\alp,\ell)^*\alp,
			\end{equation*}
			hence $y_i\in\Pa(v\alp,\ell)^*$. Since $\abs{y_i}\leq r$, all factors of $y_i$
			are at most of length~$r$, too, and $y_i\in\Pa(v\alp,r)^*$. \qedhere
	\end{enumerate}
\end{proof}

For $u\in\Sig^\ell$ let us define the language
\begin{equation*}
	\aL(u,\ell,r) = \Pa(u\alp,\ell)^* u \left( \Has(L_\alp,\ell-1,r) \cap \alp\Sig^*\ov\alp\ov u\right)
		\ov{\Pa(u\alp,r)}^*.
\end{equation*}
Note that, by induction,
for every $u$ the representation for $\aL(u,\ell,r)$ is effectively given.
By Lemma~\ref{lemon}, the word $z$ is included in $\aL(v,\ell,r)$ and for every word
$z'\in\Has(L_\alp,\ell,r)\setminus\Has(L_\alp,\ell-1,r)$
it exists
$v'\in\Sig^\ell$ such that $z'\in\aL(v',\ell,r)$. Therefore,
\begin{equation*}
	\Has(L_\alp,\ell,r) \sse \Has(L_\alp,\ell-1,r) \cup \bigcup_{u\in \Sig^\ell} \aL(u,\ell,r)
\end{equation*}
and for the right hand side we have an effective representation.
Of course, we intend to replace the inclusion by an equality sign.

\begin{lemma}\label{lem:final}
$\aL(u,\ell,r) \sse \Has(L_\alp,\ell,r)$ for all $u\in\Sig^\ell$.
\end{lemma}

\begin{proof}
	We start by proving a special case of the claim that is successfully  used 
	later to derive the result.
	Consider a word $w'$ together with the factorization
	\begin{equation*}
		w' = x_m \cdots x_1 w \ov{y_1} \cdots \ov{y_n}
	\end{equation*}
	with $m\geq 0$, $n\geq 1$ and where for some word $u\in\Sig^*$
	\begin{enumerate}
		\item $w \in \Has(L_\alp,\ell,r) \cap u\alp\Sig^*\ov \alp\ov u$,
		\item $x_1,\ldots,x_m\in\Pa(u\alp,\ell)$,
		\item $y_1,\ldots, y_n\in\Pa(u\alp,r)$, and
		\item $m = 0$ or $\abs{y_j}\leq\abs{x_m}$ for all $j$ such that $1\leq j\leq n$.
	\end{enumerate}
	We claim $w'\in\Has(L_\alp,\ell,r)$, too.
	Indeed, if $m=0$, it is plain that $w'$ is an $n$-iterated right hairpin completion of $w$.
	Otherwise $x_m\cdots x_1 w$ is an $m$-iterated left hairpin completion of $w$.
	By the fourth property and Lemma~\ref{lem:prefixes}, we have $y_1,\ldots,y_n\in\Pa(x_m\alp,r)$.
	Hence, $w'$ is an $n$-iterated right hairpin completion of $x_m\cdots x_1 w$
	and we conclude $w'\in\Has(L_\alp,\ell,r)$.
	
	\bigskip
	
	Now, let $u\in\Sig^\ell$ and $z\in \aL(u,\ell,r)$. There is a factorization
	\begin{equation*}
		z = x_s \cdots x_1 w \ov{y_1} \cdots \ov{y_t}
	\end{equation*}
	where
	\begin{enumerate}
		\item $w\in u \left( \Has(L,\ell-1,r) \cap \alp\Sig^*\ov\alp\ov u\right) \sse
			\Has(L_\alp,\ell,r) \cap u\alp\Sig^*\ov \alp\ov u$,
		\item $x_1,\ldots,x_s\in\Pa(u\alp,\ell)$, and
		\item $y_1,\ldots,y_t\in\Pa(u\alp,r)$.
	\end{enumerate}
	If $t = 0$, the word $z$ is an $s$-iterated left hairpin completion of $w$.
	Otherwise, let $n\geq 1$ be the maximal index such that $\abs{y_n}\geq\abs{y_j}$ for all $1\leq j\leq t$,
	and let $m$ be the maximal index such that $\abs{y_n}\leq \abs{x_m}$ or $0$ if no such index exists.
	Let $w' = x_m \cdots x_1 w \ov{y_1} \cdots \ov{y_n}$.
	Note that $w'$ satisfies the conditions of the special case
	we discussed above and hence $w'\in\Has(L_\alp,\ell,r)$.
	
	With $u' = y_n$ we obtain
	\begin{equation*}
		z = x_s \cdots x_{m+1} w' \ov{y_{n+1}} \cdots \ov{y_t}
	\end{equation*}
	where, by the choice of $n$, $m$ and by Lemma~\ref{lem:prefixes},
	\begin{enumerate}
		\item $w'\in\Has(L_\alp,\ell,r) \cap u'\alp\Sig^*\ov \alp\ov{u'}$,
		\item $x_{m+1},\ldots,x_s\in\Pa(u'\alp,\ell)$, and
		\item $y_{n+1},\ldots,y_t\in\Pa(u'\alp,r)$.
	\end{enumerate}
	At this point we may continue inductively and deduce $z\in\Has(L_\alp,\ell,r)$.
\end{proof}

The previous lemma tells us, if $\ell \geq r$, the \iphpc of $L_\alp$ can be represented by
\begin{equation*}
	\Has(L_\alp,\ell,r) = \Has(L_\alp,\ell-1,r) \cup \bigcup_{u\in \Sig^\ell} \aL(u,\ell,r).
\end{equation*}
Symmetrically, if $r > \ell$, let us define
\begin{equation*}
	\aR(u,\ell,r) = \Pa(u\alp,\ell)^* \left( \Has(L_\alp,\ell,r-1) \cap u\alp\Sig^*\ov\alp\right)\ov u
		\ov{\Pa(u\alp,r)}^*.
\end{equation*}
The \iphpc of $L_\alp$ can be represented by
\begin{equation*}
	\Has(L_\alp,\ell,r) = \Has(L_\alp,\ell,r-1) \cup \bigcup_{u\in \Sig^r} \aR(u,\ell,r).
\end{equation*}

We conclude, the \iphpc of a language $L$ can be represented by an expression
using $L$ and the operations
union, intersection with regular sets, and concatenation with regular sets.

\section{The size of NFAs accepting iterated parameterized hairpin completions}
\label{sec:size}

Let $L$ be a regular language and $\ell,r\in\N$ be finite bounds.
In this section we analyze the size of NFAs accepting the
iterated parameterized hairpin completion $\Hks(L,\ell,r)$
with respect to
the size of an NFA accepting $L$ and the bounds $\ell$ and $r$.
By the size of an NFA we mean its number of states.
Recall that $k$ is treated as a constant.
(Assuming $k \leq \ell$ or $k\leq r$ would induce the same complexity, but this is not shown here.)
Our results are the following.

\begin{theorem}\label{thm:size}\label{thm:nfa} \ \\ \vspace{-\baselineskip}
	\begin{enumerate}
		\item Let $m\geq 1$.
			There is a regular language $L$ such that 
			neither the language $\Hk(L,m,m)$ nor the language $\Hks(L,m,m)$
			can be detected by an NFA with less than $2^m$ states.
			\label{lower_bound}
		\item
			Let $L$ be a regular language which is accepted by an NFA of size $n$.
			Let $\ell,r\in N$ and let $m = \max\oneset{\ell,r}$.
			There is an NFA accepting the iterated parameterized hairpin completion
			$\Hks(L,\ell,r)$ whose size is in $2^{\Oh(m^2)}n$.
			\label{upper_bound}
	\end{enumerate}
\end{theorem}

\begin{proof}[Proof of \ref{lower_bound}]
	Let $\Sig=\{a,\ov a, b, \ov b, c,\ov c\}$ and $L= c\{\ov a,\ov b\}^* a^k \ov a^k$.
	For any word $w\in L$ there is no possibility of building a left hairpin
	and the only possible right hairpin is to bind the suffix $\ov a^k$ to $a^k$
	if $\abs w \leq m+2k$.
	Therefore, we have
	\begin{equation*}
		\Hk(L,m,m) = \bigcup_{v\in\{\ov a,\ov b\}^{\leq m-1}} c v a^k \ov a^k \ov v \ov c.
	\end{equation*}
	
	Now let $w = c v a^k \ov a^k \ov v \ov c$ with $v\in\{\ov a,\ov b\}^{\leq m-1}$.
	The only way to build a hairpin is to bind its prefix to its suffix, hence
	\begin{equation*}
		\Hks(L,m,m) = L \cup \Hk(L,m,m).
	\end{equation*}
	
	We claim that an NFA accepting $\Hk(L,m,m)$ or $\Hks(L,m,m)$ has a size of at least
	$2^m$.
	We prove the claim for the language $\Hk(L,m,m)$;
	the argumentation for $\Hks(L,m,m)$ is exactly the same.
	
	Consider an NFA accepting $\Hk(L,m,m)$ and let $Q$ denote its set of states.
	For a word $u\in\Sig^*$ we denote by $P(u)\sse Q$ the set of states which
	are reachable from an initial state with a path labelled by $u$.
	Now let $v\in \{\ov a,\ov b\}^{\leq m-1}$.
	Since $cva^k\ov a^k\ov v \ov c\in\Hk(L,m,m)$,
	there is a state $q\in P(cva^k\ov a^k)$ such that a path
	from $q$ to a final state exists which is labelled by $\ov v\ov c$.
	For all words $u\in \{\ov a,\ov b\}^{\leq m-1}$ with $u\neq v$
	the state $q$ does not belong to $P(cua^k\ov a^k)$ because
	$cua^k\ov a^k\ov v\ov c\notin \Hk(L,m,m)$.
	Each word $v\in\{\ov a,\ov b\}^{\leq m-1}$ yields such a state $q$,
	they are mutually different, and none of them is an initial state (as
	$\ov v\ov c\notin \Hk(L,m,m)$).
	Therefore, the number of states $\abs Q$ has to be greater than
	$\bigl| \{\ov a,\ov b\}^{\leq m-1}\bigr| = 2^m-1$.
\end{proof}

In order to prove the second claim of Theorem~\ref{thm:size}
we implicitly use some well-known constructions of NFAs
which accept concatenation, union, or intersection of regular languages.
Consider two NFAs which accept the languages
$L_1$, $L_2$ and which are of size $n_1$, $n_2$, respectively.
There is an NFA accepting the concatenation $L_1L_2$ which is of size $n_1+n_2$,
an NFA accepting the union $L_1\cup L_2$ which is of size $n_1+n_2$,
and an NFA accepting the intersection $L_1\cap L_2$ which is of size $n_1\cdot n_2$.
For details on how these NFAs are constructed see, e.g., \cite{HU}.

\begin{proof}[Proof of \ref{upper_bound}]
Let $L$ be a regular language which is accepted by an automaton of size $n$
and let $\ell,r\in \N$.
The parameterized hairpin completion of $L$ is given by
\begin{align*}
	\Hk(L,\ell,r) = &\bigcup_{\alp\in\Sig^k}\bigcup_{\gamma\in\Sig^{\leq\ell}}
			\gamma(\alp\Sig^*\ov\alp \ov\gamma \cap L) \cup
		\bigcup_{\alp\in\Sig^k}\bigcup_{\gamma\in\Sig^{\leq r}}
			(\gamma\alp\Sig^*\ov\alp\cap L)\ov\gamma.
\end{align*}
For $\gamma,\alp\in\Sig^*$ there is an NFA accepting $\gamma(\alp\Sig^*\ov\alp \ov\gamma \cap L)$
which has a size in $\Oh(\abs{\gamma\alp}\cdot n)$.
Hence, the \phpc of $L$ can be accepted
by an NFA which has a size in $\Oh(\abs{\Sig}^m m \cdot n)\sse 2^{\Oh(m)}n$
where $m=\max\oneset{\ell,r}$.

For $\alp\in\Sig^k$ the language $L_\alp = \Hk(L,\ell,r) \cap \alp\Sig^*\ov\alp$
can also be accepted by an NFA which has a size in $2^{\Oh(m)}n$.
Let $N_{i,j}$ denote the minimal size of an NFA accepting $\Has(L_\alp,i,j)$ for $i,j\in \N$.
Since $\Hk(L_\alp,0,0) = L_\alp$, we have $N_{0,0}\in 2^{\Oh(m)}n$.
For $i\geq j$ let us recall that 
\begin{gather*}
	\Has(L_\alp,i,j) = \Has(L_\alp,i-1,j) \cup \bigcup_{u\in \Sig^i} \aL(u,i,j), \\
	\aL(u,i,j) = \Pa(u\alp,\ell)^* u \left( \Has(L_\alp,i-1,j) \cap \alp\Sig^*\ov\alp\ov u\right)
		\ov{\Pa(u\alp,r)}^*.
\end{gather*}
The size of a minimal NFA accepting $\cL_\alp(u,i,j)$ is in $\Oh(i\cdot N_{i-1,j})$ whence
\begin{equation*}
	N_{i,j} \in \Oh( \abs\Sig^i i\cdot N_{i-1,j}) \sse 2^{\Oh(i)} N_{i-1,j}.
\end{equation*}
Symmetrically, for $j > i$ we have $N_{i,j} \in 2^{\Oh(i)} N_{i,j-1}$. By unfolding
the recursion we obtain
\begin{align*}
	N_{\ell,r} &\in \prod_{i=1}^\ell 2^{\Oh(i)}\cdot\prod_{j=1}^r 2^{\Oh(j)} \cdot 2^{\Oh(m)}n
		= \prod_{i=1}^m 2^{\Oh(i)} \cdot n 
		= 2^{\Oh(\sum_{i=1}^m i)}n
		= 2^{\Oh(m^2)}n.
\end{align*}
Now, the iterated parameterized hairpin completion is given by
\begin{equation*}
	\Hks(L,\ell,r) = L \cup \bigcup_{\alp\in\Sig^k} \Has(L_\alp,\ell,r).
\end{equation*}
and there is an NFA accepting $\Hks(L,\ell,r)$
which has a size in $\Oh(N_{\ell,r}+n) \sse 2^{\Oh(m^2)}n$.
\end{proof}

Statement~2 of Theorem~\ref{thm:nfa} also yields an algorithm to solve the membership
problem for the iterated bounded hairpin completion of a regular language.

\begin{corollary}
	Let $L$ be a regular language, given by an NFA of size $n$, and let $\ell,r\in\N$.
	The problem whether an input word $w$ belongs to $\Hks(L,\ell,r)$ can be
	decided in linear time $c\cdot \abs w$,
	where the constant $c$ depends on the size $n$ and the bounds $\ell$, $r$.
	More precisely, for $m = \max\{\ell,r\}$ we have $c\in 2^{\Oh(m^2)}n^2$.
\end{corollary}

\begin{proof}
	Following the proof of Statement~2 of Theorem~\ref{thm:nfa},
	we can construct an NFA $A=(Q,\Sig,E,I,F)$
	accepting the \ihpc $\Hks(L,\ell,r)$ which is of a size in $2^{\Oh(m^2)}n$.
	Let us denote the size of this NFA by $N$.
	Note that the construction can be preformed in time
	$\Oh(\abs E) \sse \Oh(N^2) \sse 2^{\Oh(m^2)}n^2$.
	
	The input $w$ can be accepted by an online power-set construction of the NFA $A$:
	We start with the set of states $P_0 = I$.
	When we read the $i$-th letter $a$ of the input $w$ we
	construct the set $P_{i}$ by following all outgoing edges of states in $P_{i-1}$
	which are labelled by $a$.
	As every state in $P_{i-1}$ has at most $N$ outgoing edges labelled by $a$,
	one step can be performed in $\Oh(N^2) \sse 2^{\Oh(m^2)}n^2$ time.
	The algorithm stops after $w$ is read and $P_{\abs w}$ is computed.
	The input $w$ belongs to $\Hks(L,\ell,r)$ if and only if
	$P_{\abs w}$ contains a final state from $F$.
\end{proof}

So far, the best known time complexity of
the membership problem for the iterated (unbounded) hairpin completion of a
regular language $L$ is quadratic with respect to the length of the input word,
by an algorithm from \cite{ManeaMM09}.
This algorithm can easily be adapted to solve the membership problem for the iterated
bounded hairpin completion in quadratic time.
Hence, if we measure the time complexity
with respect to the length of the input word only,
we have an improvement from quadratic to linear time (in the bounded case).

\section{The Iterated Hairpin Completion of Singletons}
\label{sec:ihpc}

The class of iterated hairpin completions of singletons is defined as 
\begin{equation*}
	\HCS_k = \set{\tHw}{w\in\Sig^*}.
\end{equation*}
We solve the problem whether \HCSk includes non-regular or non-context-free languages
which was asked in \cite{ManeaMY10}.
Furthermore, we will show that the result also holds if we consider the iterated one-sided \hpc.

Let us recall that, as we are treating the unbounded hairpin completion now, 
for the usual factorization $\gabag$ of a hairpin completion,
the length of the factor $\gamma$ is not bounded by a constant anymore.
By the results of the previous section it is obvious,
that the possibility of creating arbitrary long prefixes and suffixes
plays an essential role in following proof.

\begin{theorem}
	The iterated one- and two-sided hairpin completions of a singleton are 
	in \NL but not context-free, in general.
\end{theorem}

\begin{proof}
	The membership to \NL follows by the fact that \NL is closed under
	iterated bounded hairpin completion, which has been proved in \cite{ChepteaMM06}.
	For convenience, we give a sketch of the proof, here.
	
	Consider a language $L\in \NL$.
	The iterated hairpin completion $\Hks(L)$ can be accepted by a 
	non-deterministic Turing machine that works as follows.
	We use two pointers $i$ and $j$ which mark the beginning and the end
	of a factor of the input $w$, respectively.
	By $w(i,j)$ we denote the factor beginning at position $i$ and ending at
	position $j$.
	\begin{enumerate}
		\item We start with $i = 1$ and $j = \abs w$.
		\item Non-deterministically either continue with step 3 or skip to step 5.
		\item Either guess $i'$ such that $i < i' < j$ and verify that
			$w(i,j)$ is a left hairpin completion of $w(i',j)$
			or guess $j'$ such that $i < j' < j$ and verify that $w(i,j)$
			is a right hairpin completion of $w(i,j')$.
			If the verification is successful, continue with $i = i'$ (resp.\ $j = j'$).
		\item Repeat step 2.
		\item Accept if and only if $w(i,j) \in L$.
	\end{enumerate}
	
	Obviously, this Turing machine accepts $\Hks(L)$.
	In order to perform step 1-4, we only have to store some pointers
	on the input word; this can be done in $\log\abs w$ space.
	Since $L\in\NL$ step 5 can be performed in $\log\abs w$ space, too, 
	and hence $\Hks(L)\in\NL$.
	
	For the one-sided \hpc $\oneH(L)$ we can use almost the same algorithm.
	The only difference is that the pointer $i$ always is $1$.

	\bigskip
	
	Now, let $\Sig = \{a,\ov a, b, \ov b, c, \ov c\}$, $\alp = a^k$, and
	\begin{equation*}
		w = \alp b \alp \ov\alp \alp c \ov\alp.
	\end{equation*}
	We will prove that \tHw and \oHw are not context-free.
	
	Since context-free languages are closed under intersection with regular languages,
	it suffices to show for a regular language $R$ that the intersections
	$R\cap\tHw$ and $R\cap\oHw$ are not context-free. Let $u = \ov b \ov\alp$ and
	$v = \alp \ov\alp \ov b \ov\alp$.
	Note that $\ov u\alp \leq \ov v\alp \leq w$.
	Define
	\begin{equation*}
		R = w u^+ v \ov u^+ \ov w \ov u^+ \ov w
	\end{equation*}
	and consider a word $z\in R$:
	\begin{equation*}
		z = \ubr{\alp b \alp \ov \alp \alp c \ov \alp}{w}
			\ubr{(\ov b \ov \alp)^r}{u^r}
			\ubr{\alp \ov\alp \ov b \ov\alp}{v}
			\ubr{(\alp b)^s}{\ov u^s}
			\ubr{\alp \ov c \ov\alp \alp \ov\alp \ov b \ov\alp}{\ov w}
			\ubr{(\alp b)^t}{\ov u^t}
			\ubr{\alp \ov c \ov\alp \alp \ov\alp \ov b \ov\alp}{\ov w}
	\end{equation*}
	with $r,s,t \geq 1$. At first, note that $w$ is a prefix of $z$ and it does not
	occur as another factor in $z$ (there is only one $c$ in $z$).
	Thus, if $z$ belongs to \tHw, it must be an iterated right hairpin completion of $w$
	and hence
	\begin{equation*}
		R\cap\tHw = R\cap\oHw.
	\end{equation*}
	
	Next, we will show that $z$ is an iterated hairpin completion of $w$
	if and only if $r=s=t$. The proof is a straight forward construction of $z$.
	We try to find a sequence $w=w_0,w_1,\ldots,w_n = z$ for some $n\geq0$ where $w_i\neq w_{i-1}$
	is a right hairpin completion of $w_{i-1}$ for $1\leq i\leq n$.
	This implies that every $w_i$ is a prefix of $z$.

	Fortunately, for each of the words $w_0,\ldots,w_{r+1}$ there is exactly one choice
	which satisfies these conditions:
	\begin{alignat*}{2}
		w_0     & = w &&=
			\alp b \alp \ov \alp \alp c \ov \alp \\
		w_1     & = w u &&=
			\alp b \alp \ov \alp \alp c \ov\alp \ov b\ov\alp \\
		w_2     & = w u^2 &&=
			\alp b \alp \ov \alp \alp c \ov\alp (\ov b \ov\alp)^2 \\
		&\mspace{10mu}\vdots &&\mspace{10mu}\vdots \\
		w_r     & = w u^r &&=
			\alp b \alp \ov \alp \alp c \ov\alp (\ov b \ov\alp)^r \\
		w_{r+1} & = w  u^r  v &&=
			\alp b \alp \ov \alp \alp c \ov\alp (\ov b \ov\alp)^r \alp\ov\alp\ov b \ov\alp
	\end{alignat*}
	
	If $s \neq r$, none of the right hairpin completions of
	$w_{r+1}$ is a prefix of $z$ (except for $w_{r+1}$ itself).
	Otherwise, we find exactly one right hairpin completion which satisfies the conditions:
	\begin{equation*}
		w_{r+2} = w u^r v \ov u^r\ov w =
			\alp b \alp \ov \alp \alp c \ov\alp (\ov b \ov\alp)^r \alp\ov\alp\ov b \ov\alp
				(\alp b)^r \alp \ov c \ov\alp \alp \ov\alp \ov b \ov\alp.
	\end{equation*}

	The argument for the last step is the same. If and only if $t=r$, we find a prefix
	of $z$ which is a right hairpin completion of $w_{r+2}$ and this is $w_{r+3}=z$.
	
	We conclude $z$ is an iterated hairpin completion of $w$ if and only if $r=s=t$ and hence
	\begin{equation*}
		R\cap\tHw = \set{w u^r v \ov u^r \ov w \ov u^r\ov w}{r\geq 1}.
	\end{equation*}
	The intersection $R\cap\tHw$ belongs to a family of
	context-sensitive languages which are well known to be non-context-free.
	From this it follows that \tHw and \oHw are non-context-free, too.
\end{proof}

\section{Final Remarks and Open Problems}

We proved that language classes which have very basic closure properties
are closed under iterated bounded hairpin completion.
With the techniques used in our proof we obtain a better insight on the
structure of the iterated bounded hairpin completion.
This might help to design new algorithms which decide the membership of a word
to the iterated bounded hairpin completion of a given language
and also for the unbounded version since for a given word there
is an implicit given length bound.

Another interesting problem regarding the hairpin completion is whether the 
iterated hairpin completion of two languages have a common element.
Even for two given singletons it is not known, if this problem is decidable at all,
see \cite{ManeaMY10}.
The result of Section~\ref{sec:ihpc} proves that this is a non-trivial question.
However, in the bounded case we can decide this problem for two regular languages now.
We just need to create the NFAs and test whether the intersection is empty.
As the size of the NFAs is quite large with respect to the length bounds,
this does not seem to be the best way to decide the problem.

We proved the existence of non-context-free languages in  
the language class \HCS.
Here, two new questions arise naturally:
\begin{enumerate}
	\item Does a singleton exist whose \ihpc is context-free but not regular?
	\item Can we decide for a given singleton whether its \ihpc is non-regular (or non-context-free)?
\end{enumerate}


\newcommand{\Ju}{Ju}\newcommand{\Ph}{Ph}\newcommand{\Th}{Th}\newcommand{\Ch}{C%
h}\newcommand{\Yu}{Yu}\newcommand{\Zh}{Zh}

\end{document}